\documentclass[11pt]{article}

 \usepackage[linesnumbered,ruled,vlined]{algorithm2e}
 
 \usepackage[utf8]{inputenc}
\usepackage[margin=1.4cm]{geometry}
\usepackage{graphicx}
\usepackage{color}
\usepackage{amssymb,amsmath}
\usepackage{authblk}
\usepackage{tikz}
\usetikzlibrary{shapes}

\usepackage{tkz-graph}
\usepackage{multirow}
\usepackage{enumitem}
\usepackage[normalem]{ulem}
\usepackage{amsthm}

\newtheorem*{rep@theorem}{\rep@title}
\newcommand{\newreptheorem}[2]{%
\newenvironment{rep#1}[1]{%
 \def\rep@title{#2 \ref{##1}}%
 \begin{rep@theorem}}%
 {\end{rep@theorem}}}
\makeatother

\newreptheorem{theorem}{Theorem}
\newreptheorem{lemma}{Lemma}
\newreptheorem{observation}{Observation}
\newreptheorem{corollary}{Corollary}

\newtheorem{definition}{Definition}
\newtheorem{lemma}{Lemma}
\newtheorem{corollary}{Corollary}
\newtheorem{theorem}{Theorem}

\newcommand{\Dest}{t} %15 fois utilisé
\newcommand{\U}{s} %190 fois utilisé
\newcommand{\V}{V} %10 fois utilisé
\newcommand{\SSS}{\mathcal{S}} %19 fois utilisé
\newcommand{\Objection}{\mathcal{P} } %31 fois utilisé
\newcommand{\fctV}{v} %27 + 36 +160 fois utilisé

\newcommand{\ContreObjection}{\mathcal{Q} } %46 fois utilisé
\newcommand{\ShortestPath}{\mathcal{SP}} %43 fois utilisé
\newcommand{\C}{\mathcal{C}} %28 fois utilisé
\newcommand{\Cost}{c} %300 fois utilisé
\newcommand{\imputation}{x} %67+15 fois utilisé
\newcommand{\price}{p} %16 fois utilisé

\newcommand{\X}{X} %100 +15 fois utilisé

\newcommand{\owner}{owner}
\newcommand{\player}{a}
\newcommand{\announcedprice}{d}

\begin{document}

\title{Detecting service provider alliances\footnote{This work was supported by projet CAPES-COFECUB MA 828-15 CHOOSING}}
\author[1]{Johanne Cohen}
\author[2]{Daniel Cordeiro}
\author[3]{Loubna Echabbi}
\affil[1]{LRI, CNRS, Universit\'e Paris Sud, Universit\'e Paris-Saclay, France}
\affil[2]{Universidade de S\~ao Paulo, Brazil}
\affil[3]{STRS Lab., INPT, Morocco}
%
%\institute{
%LRI, Universit\'e Paris Sud, France \and
%Universidade de S\~ao Paulo, Brazil \and
%STRS Lab., INPT, Morocco
%}

\maketitle

\begin{abstract}
  We present an algorithm for detecting service provider alliances. To perform this, we modelize a cooperative
  game-theoretic model for 
  competitor service providers. A choreography (a peer-to-peer service
  composition model) needs a set of services to fulfill its 
  requirements. Users must choose, for each requirement, which service
  providers will be used to enact the choreography at lowest cost. Due
  to the lack of centralization, vendors can form alliances to control
  the market. We propose a novel algorithm capable of detecting
  alliances among service providers, based on our findings showing
  that this game has an empty core, but a non-empty bargaining set.
\end{abstract}
 
% Keywords: 
% cooperative game
% bargaining set
% algorithm design

\section{Introduction}

% ref for FI http://link.springer.com/book/10.1007/978-3-642-30241-1
The current ubiquity and pervasiveness of the Internet has been
leading researchers and practitioners to imagine the Future
Internet~\cite{future-internet}. The Future Internet, as a particular
case of Ultra-Large Scale (ULS) systems, constitutes a futuristic
vision of a yet-to-come Internet whose scale changes everything.

On such scenario, systems are usually modeled and described by a
service-oriented architecture that are completely
distributed~\cite{choreos}. Ultra-Large Scale systems will require the
transition from the current service \emph{orchestration} model for
service composition --- where the ensemble of services are composed as
an executable business process, controlled by a single party (the
orchestrator) --- to the service \emph{choreography} composition
model~\cite{orc-vs-chor} --- that describes a non-executable protocol
for peer-to-peer interactions from several different parties. In other
words, distributed scheduling devised by the applications itself
(according to their own performance objectives) will be preferred.

This new model is not only more robust and scalable, but also more
collaborative. The service choreography model enforces
interoperability and loose coupling by reflecting obligations and
constraints between parties. Each party involved will be required to
clearly state its role in a standardized manner.

Actually, some improvements on the collaboration of service providers
are already a reality for some specific services, such as Cloud
Computing platforms. Frameworks like Eucalyptus~\cite{eucalyptus},
OpenNebula~\cite{opennebula}, and OpenStack~\cite{openstack} already
allows developers to choose different cloud providers without any
additional change on their software. A cloud computing user can choose
(at any time) one or more providers from a large set of options. The
choice can be driven by different non-functional criteria such as
total cost, quality-of-service, etc.

Beyond the complexities of the management of composite services
(design, provisioning, etc.), we consider the problem from an economic
point of view. Since service vendors are not regulated, both
cooperative and non-cooperative behavior may be expected. This can
potentially lead to the formation of alliances --- i.e., the formation of
groups of similar independent parties, who join together to control
prices and/or limit competition.

% Ground works from web semantics~\cite{semantic2,semantic3} and
% software engineering~\cite{choreos,more-choreos} are creating a
% concrete framework for making the Future Internet scenario a
% reality. From an economic point of view, independent (and selfish)
% service providers could potentially create alliances to control the
% market.

In this work we use tools and techniques from cooperative game
theory~\cite{courcoubetis} in order to analyze the stability and the
alliance formation on the Future Internet scenario; a
problem we call the \emph{choreography enactment pricing game}.

\section{Problem statement}
\label{sec:problem-statement}

Providing different types of services to end-users is a very promising
business model. Companies can provide access to complex systems as a
service, where end-users can pay to execute an operation. It can scale
from small companies (e.g., a travel company that sells airline
tickets for end-users) --- to large companies that offers a huge
%(quasi-infinite)
number of computational resources (like public cloud
computing providers, that provide their services in a pay-as-you-go
manner.)

In order to execute a complex choreography, a user must select a set
of service vendors that commercialize the services that fulfills all
choreography requirements. A same type (or role) of service can be
offered from different vendors at different prices. Also, the same
vendor can commercialize more than one type of service. It is the
end-user that must choose (either manually or on an automated manner)
which service provider will be in charge of the execution of services
that fulfills each role. This process is known in systems literature
as \emph{choreography enactment}~\cite{cordeiro14nca}.

We use a graph of precedence constraints to model all possible ways to
enact a choreography. The vertices of the graph indicate all known
public services that fulfill the needs of the user and the edges
represents the order in which those services must be executed.

More formally, all possible combinations of services that can enact a
given choreography is modeled as a directed acyclic graph
$G=(\V \cup \{\U, \Dest\},\, E)$, such that $w \in \V$ is a service
that users can pay to use, and the existence of an edge
$\{w, x\} \in E$ indicates that the service $w$ must be executed
before $x$. Let $m$ be the number of public services available to the
user.

Vertices $\mathbf{\U}$ and $\mathbf{\Dest}$ are artificial vertices added to represent,
respectively, the \emph{start} and the \emph{end} of the choreography. Thus, for all
services $x$ with in-degree equal to zero, we have an additional edge
$\{\U, x\}$. Similarly, for all services $w$ with out-degree zero, we
also have an edge $\{w, \Dest\}$.

A successfully enacted choreography can be represented by a path from
$\U$ to $\Dest$ on the graph $G$.
 
   \begin{figure}
 \begin{center} \begin{tikzpicture} %
     \node[draw,rectangle] (4)  at (2,2) {Cloud Provider 2};
     \node[draw,rectangle] (2) [above of=4] {Cloud Provider 1};
      \node[draw,rectangle] (3) [below of=4] {Cloud Provider 3};
     \node[draw,circle]  (1) at (-2,2)   {$\U$};
    \node[draw,rectangle] (5) at (6,3)  {Database};
     \node[draw,rectangle] (6)  at (6,1)  {Storage service};
     \node[draw,circle]  (7) at (10,2)  {$\Dest$};
     \draw[->] (1) -- (2);
     \draw[->] (2) -- (5);
     \draw[->] (5) -- (7);
     \draw[->] (1) -- (3);
     \draw[->] (3) -- (6);
     \draw[->] (6) -- (7);
     \draw[->] (4) -- (5);
     \draw[->] (1) -- (4);
     \draw[->] (4) -- (6);
 \end{tikzpicture} \end{center}
  \caption{Choreography enactment pricing game with empty core.}
  \label{fig:simple-example}
\end{figure}
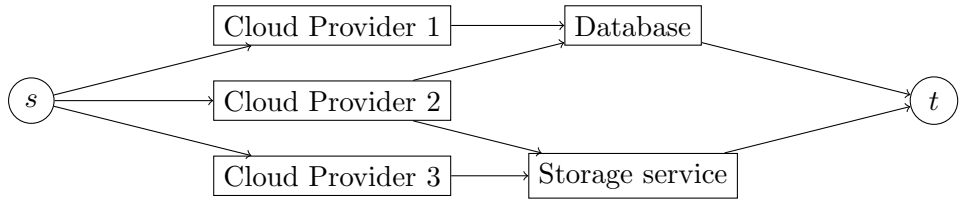

Figure~\ref{fig:simple-example} shows a graph depicting different forms of enacting a given choreography. In this case, the user can choose among three different \emph{cloud providers} and two data serialization services. Any path $\U$--$\Dest$ is valid, but a user would choose the cheaper one.

Users must pay a price to the service provider in order to use
its services. Each service node $w$ charge $\Cost_w$ for its use. This
value is fixed and known in advance by all providers. We define the
costs of vertices $\U$ and $\Dest$ as being zero.

The total price paid by a user is straightforward defined as the sum
of the costs of the vertices in the path $\U$--$\Dest$ chose by the
user. Let $\Cost$ be the vector $(\Cost_1, \dots , \Cost_m)$ of all
services costs and $\Cost^{-k}$ the vector
$(\Cost_1, \dots, \Cost_{k-1}, \Cost_{k+1}, \dots, \Cost_m)$ of costs
for all services except service $k$.

Users also have a budget, i.e., to set the maximal price $\price$ that
can be payed to realize an enactment. We assume that a user will
always choose the path with minimal cost. If this minimal cost is
greater than $\price$, then the user simply refuses to execute its
choreography.

From the service vendors point of view, we have a set of services that
are offered to users at some given price (its operating costs plus a
profit). A set of service vendors $\SSS$ compete among each other to offer those services. We use the function $\owner: \V \mapsto \SSS$ to denote the mapping from a given service $w$ to
its vendor, $\owner(w)$.

In a free market, however, nothing prevents vendors from creating
alliances. Those alliances could be used to control the market by
fixing the smallest price that would lead a user to choose their own
services (\emph{price wars}). The price can be decreased as long as
the total operating cost (the cost to execute the services) of the
coalition can be covered by the price payed by the user, and the
profit can then be split between the members of the alliance --- a
concept known in game-theory as \emph{transferable utility}.

To formalize these alliances (or coalitions), we use the notion
of \emph{cooperative games}. We assume the following hypothesis for
our model:

\begin{itemize}
\item a path (a combination of services) is offered to the user
  at a unique price;%  Service vendors in the same path can use a
  % revenue sharing mechanism to transfer profits from one service to
  % other. It allows providers to redistribute the revenue among them in
  % order to allow a better offer with smaller price. In game theory, this
  % class of games is known as games with \emph{transferable utility}.

\item a user always selects a path with smaller cost;%  Without any
  % coordination, this free market allows \emph{price wars}, where
  % vendors will redistribute profits among them in order to decrease
  % the total announced price.

% \item Coalitions with providers other than
%   those required to enact the choreography may actually increase the
%   overall profit. The idea is that those additional vendors can
%   neutralize competitive alternatives paths choosing appropriate
%   prices for its services.

\item the revenue sharing mechanism must ensure \emph{coalition
  stability}. In other words, providers are never tempted to leave the
  coalition; % and their path is guaranteed to be cheapest one.

\item the coalition must offer at least one valid path, with at least two services, 
  composed only by providers belonging to that coalition --- a unitary path 
%
% \item A path cannot be comprised of only one service. Otherwise, it
  would create a monopoly controlled by its owner;

\item the coalition can profit only if its path is guaranteed to be
  selected. Either there is no other valid path costing lesser than
  the maximum price $\price$, or other paths are definitely more
  expensive.% Alternative competitors are thus
  % neutralized. The offered total price can be lower
  % than any other path (if it exists) and even $\price$, otherwise.
\end{itemize}

Before describing the game that models those properties, we will
recall some important concepts from game theory that will be used in
this paper.

\section{Cooperative game theory \cite{coalitiongame}}%

\subsection{Transferable utility games and Characteristic functions}

A \emph{cooperative game with transferable utility} (TU game) is a
pair $(\SSS, \fctV)$ where $\SSS = \{1, \dots, n\}$ is a finite set of
players and a \emph{characteristic function}
$\fctV: 2^{|\SSS|} \to \mathbb{R}$ which associates each subset
$\C \subseteq N$ to a real number $\fctV(\C)$ (such that
$\fctV(\emptyset) = 0$). Each subset $\fctV(\C)$ of $\SSS$ is called a
\emph{coalition}. The function $\fctV$ is a \emph{characteristic
  function} of the game $(\V, \fctV)$ and the \emph{value} of
coalition $\C$ denoted by $\fctV(\C)$ is the value that $\C$ could
obtain if they choose to cooperate. In TU games, the value of a
coalition can be redistributed among its members in any possible way.

\subsection{Revenue sharing mechanism}

The challenge of a revenue sharing mechanisms is to find how to split the
payoff $\fctV(\C)$ among the players in $\C$ while ensuring the
stability of the coalition. A vector
$x=(\imputation_1,\dots,\imputation_{|\V|})$ is said to be a
\emph{payoff vector} for a $k$-coalition $\C_1, \dots, \C_k$ if
$\imputation_i \leq 0$ for any $i \in \V$ and
$\sum \limits_{i\in \C_j} \imputation_i \leq \fctV(\C_j)$ for any
$j \in \{1, \dots, k\}$. We will focus on some particular payoff vectors,
namely \textit{imputations}.

\begin{definition}
A payoff vector $\imputation$ for a $k$-coalition $\C_1, \dots, \C_k$ is said to be an \emph{imputation} if it is
\emph{efficient}, --- i.e., $\sum \limits_{i\in \C_j} \imputation_i \leq \fctV(\C_j)$
for any $j \in \{1, \dots, k\}$ --- and if it satisfies the individuality
rationality property --- i.e., $\imputation_i \geq \fctV(\{j\}) $ for any player $j \in \SSS$.
\end{definition}

The objective is to find a \emph{fair}
distribution of the value of the coalition (the payoff of each player
corresponds to his actual contribution to the coalition) and also to
ensure a \emph{stable} coalition in such a way that no player or
subset of players have incentive to leave the coalition.

% \paragraph{The Shapley value.}

% The \emph{Shapley value} \cite{Shapley} is one way to distribute
% the value $\fctV(\N)$ to all players , assuming that they all
% collaborate.  The \emph{Shapley value} is based on the following idea
% : the payment that each player receives is proportional to ``his
% contribution'' (he receives the payment equals to the increasing
% value of the coalition when he joins it). According to the Shapley
% value, the amount that player $\ell$ gets given a cooperative  game $(\N,\fctV) $ is

% $$\displaystyle \valeurShapley_\ell(\fctV)=\sum_{S \subseteq \N\setminus
% \{\ell\}} \frac{|S|!\; (|\N|-|S|-1)!}{|\N|!}(\fctV(S\cup\{\ell\})-\fctV(S))$$

% The Shapley value provides a fair way to share the value of a coalition but do not consider stability issues.

\paragraph{The Core}\cite{core}

Let $(\SSS, \fctV)$ be a cooperative game and $\imputation$ a payoff
vector of this game.  A pair $(\Objection, y)$ is said to be
\emph{objection of $i$ against $j$} if:

\begin{itemize}
\item $\Objection$ is a subset of $\SSS$ such that $i \in \Objection$
  and $j \notin \Objection$ and

\item if $y$ is a vector in $\mathbb{R}^\SSS$ such that
  $y(\Objection) \leq \fctV(\Objection)$, for each
  $ k \in \mathcal{P}$, $y_k \geq \imputation_k$ and
  $y_i > \imputation_i$ (agent $i$ strictly benefits from $y$, and the
  other members of $\Objection$ do not do worse in $y$ than in
  $\imputation$).
\end{itemize}

In this case, we say that Imputation $\imputation$ is \emph{dominated} by $y$.

\begin{definition}
  \label{core_def}
  \emph{The core} is the set of imputations for which there is no
  objection. That is imputations that are not dominated.

  Another equivalent definition,  states that the core is a set of payoff allocations $x\in\mathbb{R}^\SSS$ satisfying:
  \begin{itemize}
  \item Efficiency: $\sum\limits_{i\in \SSS}\imputation_i=\fctV(\SSS)$;
  \item Coalitional rationality: $\sum\limits_{i\in \C} \imputation_i\geq \fctV(\C), \forall \C \subseteq \SSS$.
  \end{itemize}
\end{definition}

When the core is nonempty the grand coalition, the coalition
created by all service providers, $\SSS$ is stable as there is no
objection threatening to leave
it.  % We now go back properties that drive interest towards convex games.
  %       \begin{theorem}\cite{core}
  %       The core of a convex game is non empty.
  %       \end{theorem}

  % \begin{theorem}\cite{core}
  %       In convex games Shapley value is the center of gravity of  the core.
  %       \end{theorem}

The core is a strong concept. Sometimes the core is empty and this
requirement should be relaxed. The following concept considers
situations where objections are not justified and may be
\textit{neutralized}.

\paragraph{The bargaining set.}\cite{Bargainingset,BargainingsetObjection}

A pair $(\ContreObjection, z)$ is said to be a \emph{counter-objection} to an objection
$(\Objection, y)$ if:

\begin{itemize}
\item $\ContreObjection$ is a subset of $\SSS$ such that
  $j \in \ContreObjection$ and $i \notin \ContreObjection$ and

\item if $z$ is a vector in $\mathbb{R}^\SSS$ such that
  $z(\Objection) \leq \fctV(\Objection)$, for each
  $k \in \ContreObjection \backslash \Objection$,
  $z_k \geq \imputation_k$ and, for each
  $k \in \ContreObjection \cap \Objection$, $z_k \geq y_k$ (the
  members of $Q$ which are also members of $\Objection$ get at least
  the value promised in the objection).
\end{itemize}

Let $(\SSS, \fctV)$ be a game with a coalition structure. A vector
$x \in \mathbb{R}^\SSS$ is \emph{stable} iff for each objection at
$\imputation$ there is a counter-objection.

\begin{definition}
  The bargaining set $\mathcal{B}((\SSS, \fctV))$ of a cooperative game
  $(\SSS, \fctV)$ is the set of stable payoff vectors that are
  \textit{individually rational}, that is,
  $\imputation_i \geq \fctV(\{i\})$.
\end{definition}

Note that it is sufficient to use the notion of imputations since the
payoffs are individually rational.

The bargaining set concept requires objections to be immune to
counter-objections, otherwise they are not considered as credible
threats.

Now, we will go back to our problem and try to define appropriate
characteristic function that reflects the outcome expected from
coalition formation as described earlier in the problem statement
section (Section~\ref{sec:problem-statement}). Note, however, that
even if the core is empty, the bargaining set is not.

\section{The choreography enactment pricing game}

% Recall that in our composite service provisioning game the set of players is  $\SSS=\Sb \bigcup \Ssub$ ; the set of service providers including broker service provider and sub-service providers.\\

The \emph{choreography enactment game} models the cooperative game played by service providers. Their main objective is to form coalitions in order to create the best (cheaper) choice for a user that wants to deploy (enact) a given choreography.
The utility (profit) of a service provider 
is given as a function of the
price $\announcedprice_w$ that the vendor announced for service $w$
--- if this coalition is chosen by the user --- and the operating cost of
this service. Formally, the utility $\fctV(w)$ received by service $w$ (vertex in graph $G$) % Modification de Johanne
is given by:

\begin{equation}
\label{eq:fctV:vertex}
 \fctV(w) = \begin{cases}
    \price_w - \announcedprice_w & \text{if } w \text{ is in the chosen path} \\
                   0 & \text{otherwise} \end{cases}
\end{equation}

% ajout de Johanne
Now, from the utility $\fctV_w$ received by service $w$, we derive the    utility $\fctV({\player})$ received by player $\player$ (set of vertices in graph $G$):

\begin{equation}
\label{eq:fctV}
 \fctV(\player) = \displaystyle \sum_{w \in \V : \owner(w) = \player}  \fctV(w)
 \end{equation}

% Fin de l'ajout de  Johanne

We can generalize this notation for a coalition. If $X$ is a set of service providers forming a coalition, its imputation can be given by:
\begin{equation}
\label{eq:fctV-coalition}
\fctV(X) = \sum_{\player \in X} \fctV(\player) % Modification de Johanne
\end{equation}

We note the path leading to the lowest cost offered by $\X$ as
$\ShortestPath^{\X}$ and the cost of the shortest path as
$\Cost^{\X}$. Similarly, we note the shortest
path/lowest-cost $k$-avoiding path from $\U$ to $\Dest$ as
$\ShortestPath^{-k}$ and its total cost as $\Cost^{-k}$. If no such
path exists, its cost is defined to be $\infty$.

% \begin{definition}
%   Let $\mathcal{G}$ be a game. Let $\X$ be a coalition. We define the
%   characteristic function $\fctV(.)$ as following:

%   \begin{itemize}
% \item[]    $\fctV(\X) =   min(\price, \Cost^{-\X}) -   \Cost^{\X} $  if $\X$ contains a path between nodes $\U$ to $\Dest$,  \\          \hspace{4cm} and if $\Cost^{\X}  < \Cost^{-\X}$,
%  \item[] $\fctV (\X)= 0$
%  otherwise
%  \end{itemize}
%   \end{definition}

This characteristic function satisfies the properties stated in
Section~\ref{sec:problem-statement}. A coalition must offer at least
one valid path, otherwise its value is 0. The coalition may have some
profit only if its path is guaranteed to be selected. That is, outside
the coalition, there is no valid path or the other paths are more
expensive --- i.e., $\Cost^{\X} < \Cost^{-\X}$. Thus, the offered
price must be better than the total cost of the shortest path outside
the coalition ($\Cost^{-\X}$), if it exists, or the maximum price that
the user is ready to pay ($\price$) otherwise. We can derive that:

\begin{equation}
\label{eq:fctV:vertex}
 \fctV(w) = \begin{cases}
    \min(\price, \Cost^{-\X}) - \Cost^{\X} & \text{if } \Cost^{-\X} \ge \Cost^{\X} \\
                   0 & \text{otherwise}
            \end{cases}
\end{equation}

\begin{corollary}
  The value of the grand coalition $\SSS$ in the choreography
  enactment game is given by:
  $\fctV(\SSS)=(\price-\Cost_{\ShortestPath})$
\end{corollary}

This result is straightforward. Since there is no path outside the
grand coalition, we have that $min(\price, \Cost^{-\SSS}) = \price$.
The lowest cost path inside $\SSS$ is necessarily the absolute
shortest path $\ShortestPath$.

We will see later that, in the general case, it is not necessary to be
in the grand coalition to reach $\price$ as the cost of the coalition.

\paragraph{\bf Example:}
  Consider the following example, illustrated in
  Figure~\ref{fig:contre_example}. In this example, different vertices'
  shapes represents different owners. Player $\Lambda$ owns two
  vertices, $\alpha$ and $\lambda$, while others own only one vertex
  each. For the sake of simplicity, the set $\SSS$ of players is
  $\{\Lambda, \beta , \delta, \gamma\} $.  We assume that the user
  budget is $\price= 34$.  Since the cost $\Cost_{\ShortestPath}=6$,
  $\fctV(\SSS)=\price - \Cost_{\ShortestPath} = 28$.

  To compute the value $v(\{\Lambda, \gamma\})$, we must compute the
  difference between the cost of the shortest path not containing $\{\Lambda,\gamma\}$, which is
  equal to 23, and the shortest path containing only $\{\Lambda,\gamma\}$,
  which is equal to 6.
 
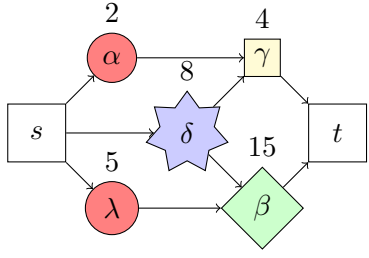
\begin{figure}
 \centering
\begin{tikzpicture}
 %                   \tikzstyle{p}=[rectangle,draw,minimum size=2em]
                   \tikzstyle{p1}=[circle,draw,fill=red!50]
                       \tikzstyle{p2}=[diamond,draw,fill=green!20]
                       \tikzstyle{p3}=[star,star points=7,draw,fill=blue!20]
                       \tikzstyle{p4}=[draw,fill=yellow!20]
 
\node[rectangle,draw,minimum size=2em] (1) at (0, 1) {$\U$};
\node[circle,draw,fill=red!50] (2) at (1, 2)    {$\alpha$};
\draw (2.north)    node[above]{$2$};
%%\node[circle,draw,fill=red!50] (2) at (1, 2) [label=$2$]  {$\alpha$};
   \node[p1] (3) at (1, 0)   {$\lambda$};% \node[p1] (3) at (1, 0) [label=$5$]  {$\lambda$};
\draw (3.north)    node[above]{$5$};
  \node[p3] (4) at (2, 1)  {$\delta$}; \draw (4.north)    node[above]{$8$};
   \node[p4] (5) at (3, 2)    {$\gamma$}; \draw (5.north)    node[above]{$4$};
   \node[p2] (6) at (3, 0)  {$\beta$}; \draw (6.north)    node[above]{$15$};
    \node[rectangle,draw,minimum size=2em]   (7) at (4, 1) {$\Dest$};
    \draw[->] (1) -- (2);
    \draw[->] (2) -- (5);
    \draw[->] (5) -- (7);
    \draw[->] (1) -- (3);
    \draw[->] (3) -- (6);
    \draw[->] (6) -- (7);
    %\draw[->] (3) -- (4);
    \draw[->] (4) -- (5);
    \draw[->] (1) -- (4);
    \draw[->] (4) -- (6);
 \end{tikzpicture}
 \caption{Choreography enactment pricing game with empty core.}
 \label{fig:contre_example}
\end{figure}

 We now show that the core of this example is empty.
Its characteristic function is given by:

\[
\begin{array}{ll}
\fctV(\{\Lambda,\delta,\gamma,\beta\}) = 34-6 = 28;    &\fctV(\{\Lambda,\gamma,\delta\}) = 34-6 = 28;   \\
  \fctV(\{\Lambda,\gamma\})=23-6=17;       &\fctV(\{\Lambda,\gamma,\beta\}) = 34-6 = 28;   \\
\fctV(\{\gamma,\delta\}) = 20-12 = 8;   &\fctV(\{\delta,\gamma,\beta\}) = 34-6 = 28;  \\
\fctV(\{\Lambda,\beta,\delta\}) = 34-20 = 14;   &
\end{array}
\]

For all other coalitions the value must be zero; either it does not contain
a valid path or it can be beaten by a shortest path outside the
coalition.

Suppose, by contradiction, that there exists an imputation
$\imputation=(\imputation_\Lambda, \imputation_\gamma, \imputation_\delta, \imputation_\beta)$ that
belongs to the core.  By Definition~\ref{core_def}, from the
efficiency property, we have that
$\imputation_\Lambda + \imputation_\gamma + \imputation_\delta + \imputation_\beta= 28$.
From the group rationality property of the coalition $\{\Lambda,\gamma,\delta\}$, we
have that $\imputation_\Lambda +\imputation_\gamma+ \imputation_\delta  \geq 28$, so
definitely $\imputation_\beta = 0$. Applying the same argument on
coalition $\{\Lambda,\gamma,\beta\}$ results that $\imputation_\delta=0$. Substituting
$\imputation_\delta$ and $\imputation_\beta$ using the group rationality
property on coalition $\{\Lambda,\beta,\delta\}$ results that
$\imputation_\Lambda \geq 14$. Similarly, on coalition $\{\delta,\gamma,\beta\}$, it
yields $\imputation_\gamma,\geq 28$. Thus
$\imputation_\Lambda, +\imputation_\delta \geq 28 + 14 = 42$ which contradicts the
efficiency property. Therefore, the core of this game is empty.

From this example we can notice that both coalitions $\{\Lambda,\gamma,\delta\}$ and
$\{\Lambda,\gamma,\beta\}$ have the same value of the grand coalition. Nodes $\delta$ and
$\beta$ both play the same role, but are still important to the
coalition. In the core, their imputation is zero, by symmetry. Still,
they are important because other mechanisms could be used in order
to redistribute them according to their contributions. In the
general case that a cut of the graph (including, or in addition to the
shortest path) are sufficient to get the same value as the grand
coalition. As we have stated above when the core in empty, it is still
interesting to consider the Bargaining set.

% Johanne a modifie la section suivante
\section{Stable coalitions on the choreography enactment pricing game}

The commerce of services is not regulated. This free market implies
that service vendors are free to create alliances that could
potentially allow them to control the entire market.  In the previous
section we saw that the core of the choreography enactment pricing
game is empty. In this section we show that the bargaining set is
non-empty and we derive an algorithm that is capable of detecting
the formation of such alliances.

In graph theory, a \emph{vertex cut} $\C$ for vertices $\U$ and
$\Dest$ is a set of the vertices such that its removal from graph $\C$
separates $\U$ and $\Dest$ into distinct connected components.  We
focus on a \emph{vertex cut} $\C$ and some useful properties
about function $v(.)$.

\begin{lemma}
  Let $\C$ be a vertex cut in graph $G$. Let $\ShortestPath$ be the
  shortest path between nodes $\U$ and $\Dest$.  Let
  $\C_{\ShortestPath}= \C \cup \ShortestPath$ be a set of vertices. Let consider
  $w \notin \C_{\ShortestPath}$ a vertex. We  have the following property for the vertex's utility function:
  $$\fctV(\C_{\ShortestPath})=\fctV(\C_{\ShortestPath} \cup\{w\})$$
\end{lemma}

\begin{proof}
  It is sufficient to notice that outside $\C_{\ShortestPath}$ there
  is no path between nodes $\U$ and $\Dest$ in graph $G$ thus
  $\fctV(\C_{\ShortestPath})=\fctV(\V)$.  \qed
\end{proof}

This implies in the following property about the value function:

\begin{lemma}
  Let $\X$ be a set of players such that
  $\C_{\ShortestPath} =\{w: \owner(w) \in \X\}$ contains a vertex cut
  $\C$ and a be the shortest path $\ShortestPath$ between nodes $\U$
  and $\Dest$. Let us consider $\player \notin \X$ a player. Then, the
  following property for the value function holds:
  $$\fctV(\X)=\fctV(\X \cup\{w\})$$
\end{lemma}

\begin{lemma}
  \label{lem:simple}

  Let $\mathcal{G}(G,\fctV,\price)$ be a game. Let $\imputation$ be a
  feasible stable imputation. For each player $j$ in $\SSS$ such that
  $ \Cost^j \geq \price$, we have $\imputation_j=0$.
\end{lemma}

\begin{proof}
  We prove this lemma by contradiction. Assume that there exists a
  player $j$ in $\SSS$ such that $\Cost^j \geq \price$ and
  $\imputation_j > 0$.

  There exists at least a player $i$ such that $ \Cost^i \leq \price$ (otherwise, the user would not choose any path since all prices of the path are greater than $\price$: $\imputation_j=0$ for all $j \in \SSS$).   Therefore
  $\Cost_{\ShortestPath} \leq \price$ and there exists at most one vertex $w$ such
  that $\Cost^w = \Cost_{\ShortestPath}$. So player $\owner(w)$ has  $\Cost^{\owner(w)} = \Cost_{\ShortestPath}$.

  Let $i^*$ be a player such that $\Cost^{-i} \geq \Cost^{-k}$ for each
  $k\in\SSS$ and such that $\Cost^{k}= \Cost_{\ShortestPath}$.
  First,  Equations~\ref{eq:fctV}
  and~\ref{eq:fctV-coalition} state that, for a given coalition $\X$, we have
   $\fctV(\X) > 0 $ if $ \Cost^{\X} < \Cost^{-\X}$,  otherwise
    $\fctV(\X) = 0$.

 Since $\Cost^{j} \geq \price$, each coalition $\X$ satisfies
 the following property: $\fctV( \X) = \fctV(\X\backslash\{j\})$.

 Player $i^*$ could make an objection of $(\SSS\backslash\{j\},y)$
 against node $j$ such that
 $y_k=\imputation_k+\frac{\imputation_j}{|\SSS|-1}$ for
 $k\in \SSS\backslash\{j\}$. Note that for any $k$, $y_k > x_k$ since
 $x_j>0$. Now consider two cases:

 \begin{enumerate}
 \item $\mathbf{\Cost^{-i^*} = \Cost_{\ShortestPath}}$. This implies that
   $\fctV(\SSS\backslash\{j\}) = 0$ and
   $\fctV(\SSS\backslash\{i^*\}) = 0$.
 \item $\mathbf{\Cost^{-i^*} > \Cost_{\ShortestPath}}$. This implies
   that $\Cost^{(\SSS\backslash\{i^*\})}> \Cost_{\ShortestPath}$. As
   consequence of the definition of $\fctV(.)$, we have
   $\fctV(\SSS\backslash\{i^*\}) =0$.
 \end{enumerate}

 Thus, $\fctV(\SSS\backslash\{i^*\}) = 0$. Player $j$ cannot make a
 counter-objection $(\ContreObjection,z)$ against $i$ because
 $\fctV(\ContreObjection\backslash\{i^*\}) = 0$ for all coalitions
 $\ContreObjection$ not containing $i^*$.  This contradicts the
 assumption that $\imputation$ is stable.  \qed
\end{proof}

Let's focus on service providers that will receive non null
retribution.

\begin{lemma}\label{lem:6.3}
  Let $\mathcal{G}(G,\fctV,\price)$ be a game.  Let $\imputation$ be a
  feasible imputation.  Let $i$ (similarly $j$) a player such that
  $ \Cost^i < \price$ (similarly $ \Cost^j < \price$).  Let
  $\mathcal{O}=\SSS\backslash\{j\}$.

  Let $(\mathcal{O},y)$ be an
  objection of $i$ against $j$. In order to have a counter-objection
  to $(\ContreObjection, z)$, with
  $\ContreObjection= \SSS\backslash\{i\}$ of $j$ against $i$, a
  sufficient condition is:
\begin{equation}
  \label{eq:relation}
  \imputation_j -\imputation_i  \leq  \left(\Cost^{-i}-\Cost^{-j}\right)
\end{equation}
 \end{lemma}

\begin{proof}
  Assume that there is an objection $(\mathcal{O},y)$ of player $i$
  against player $j$. By definition of objection, we have that
  $\forall k \in \mathcal{O}$, $y_k \geq \imputation_k$ and
  $y_i > \imputation_i$. So $0< \fctV(\mathcal{O})$ and thus
  $\Cost^{\mathcal{O}}\leq \Cost^{-\mathcal{O}}$ by the definition of
  the value of a coalition
  ($\fctV(\mathcal{O}) = \min
  (\price,\Cost^{-\mathcal{O}})-\Cost^{\mathcal{O}}$).

  The value of coalition $\mathcal{O}$ depends on the lowest-price
  $j$-avoiding path from $\U$ to $\Dest$:
  $\Cost^{\mathcal{O}} =\Cost^{-j} $. Therefore:

  \begin{equation} \label{eq:3:bis}
    \Cost^{-j}\leq \min (\price,\Cost^{-\mathcal{O}}).
  \end{equation}

  Now, we will look for a counter-objection of player $j$ using
  $(\ContreObjection, z)$, where
  $\ContreObjection= \SSS\backslash\{i\}$. Its value of coalition
  depends on the lowest-price $i$-avoiding path from $\U$ to $\Dest$:
  $ \Cost^{\ContreObjection} =\Cost^{-i}$. Since
  $\fctV(\ContreObjection) =\min
  (\price,\Cost^{-\ContreObjection})-\Cost^{\ContreObjection}$,
  we have:

 \begin{center}
   $ \Cost^{\ContreObjection}\leq \Cost^{-\ContreObjection}$ and
   $\fctV(\ContreObjection) =\min
   (\price,\Cost^{-\ContreObjection})-\Cost^{-i}$
 \end{center}

 By combining the two previous equations with Equation~\eqref{eq:3:bis}, we obtain:

 \begin{equation}
   \label{eq:inter:1}
   \begin{array}{rl}
     \fctV(\ContreObjection) - \fctV(\mathcal{O})  &= \min (\price,\Cost^{-\ContreObjection})-\Cost^{-i} -(\min (\price,\Cost^{-\mathcal{O}})-\Cost^{\mathcal{O}} )\\
     \fctV(\ContreObjection) - \fctV(\mathcal{O})  &\leq \left(\Cost^{-j}- \Cost^{-i} - \min (\price,\Cost^{-\mathcal{O}})+\min (\price,\Cost^{-\ContreObjection}) \right)
   \end{array}
 \end{equation}

From the definition of the game, player $i$ (resp. $j$) cannot be
   simultaneously adjacent to $\U$ and $\Dest$ (otherwise a monopoly would be possible.) This implies that
 $\min (\price,\Cost^{-\mathcal{O}})=\min
 (\price,\Cost^{-\ContreObjection})$.

 Let $N$ be a set of vertices such that $N =\SSS \backslash\{i,j\}$
 Since $\fctV(\mathcal{O}) = \sum_{k=1}^{|N|}y_k +y_i$ and
 $\fctV(\ContreObjection) = \sum_{k=1}^{|N|}z_k + z_y$,
 Equation~\eqref{eq:inter:1} can be rewritten as:

 \begin{equation}\label{eq:inter:2}
   \sum\limits_{k\in \ContreObjection}z_k -\sum\limits_{k\in \mathcal{O}}y_k  \leq
   \left(\Cost^{-j}- \Cost^{-i}\right)
 \end{equation}

 Let us now show that $\ContreObjection$ is an counter-objection
 where, for each $k \in \ContreObjection\backslash \mathcal{O}$,
 $z_k \geq \imputation_k$ and for each $k \in S\backslash\{i,j\}$,
 $z_k \geq y_k$. It is sufficient to consider that:

 \begin{equation}\label{eq:inter:3}
   z_j- y_i  \leq  \left(\Cost^{-j}- \Cost^{-i}\right)
 \end{equation}

 From the definition of objection, it is sufficient to have:
 $\imputation_j -\imputation_i \leq \left(\Cost^{-i}-
   \Cost^{-j}\right)$. This concludes the proof of the lemma. \qed
\end{proof}

Lemma~\ref{lem:simple} shows that a feasible imputation $\imputation$ is stable if,
for each player $j$ in $S$ such that $\Cost^j \geq \price$, we have
$\imputation_j=0$. We will now compute the value of
$\imputation_j$ for all players $j$ such that $\Cost^j \geq \price$. We
will define the set  of players $A$ such that
$A=\{j\in \SSS: \Cost^j \leq \price\}$.

\begin{theorem}\label{th:5}
  Let $\mathcal{G}(G,\fctV,\price)$ be a game.  Let $A$ be a subset of players $\{j\in \SSS: \Cost^j \leq \price\}$. There exists a unique stable imputation $\imputation$ if $\imputation$ fulfils all the three following conditions:
  \begin{enumerate}
  \item $ \forall j \in \SSS\backslash A $, $\imputation_j=0$;
  \item
    $ \displaystyle \forall j \in A, \imputation_j =
    \frac{\fctV(\SSS)}{|A|} + \left(\Cost^{-j}-\frac{ \displaystyle
        \sum_{k\in A}\Cost^{-k}}{|A|} \right)$; and,
  \item
    $\forall j \in A, \Cost^{-j} \geq \frac{\Cost_{\ShortestPath} -
      \price + \displaystyle \sum_{k\in A}\Cost^{-k}}{|A|}$.
  \end{enumerate}

\end{theorem}

\begin{proof}
  Property (1) can be straightforward deduced from Lemma~\ref{lem:simple}.

  The bargaining set is the set of all imputations that do not admit a
  justified objection. So, if we apply Theorem~\ref{th:5} to $i,j$ and
  then to $j,i$, then we can derive that for any couple $(i,j) \in S^2$, we
  have $\imputation_i - \imputation_j = (\Cost^{-i} - \Cost^{-j})$.

  Let $j$ be a player in $A$. By summing all the previous equations, we
  obtain:
  \begin{equation}\label{eq:inter:4}
    \sum_{k\in A }\imputation_i -|A|\imputation_j  = \left(\sum_{k\in A} \Cost^{-k}-|A|\Cost^{-j}\right)
  \end{equation}

  From the properties of the value of the coalition and by computation, we can rewrite Equation~\eqref{eq:inter:4} as:

  \begin{equation}\label{eq:inter:5}
    \forall j \in A, \; \imputation_j  =  \frac{\fctV(\SSS)}{|A|} +
    \left(\Cost^{-j}-\frac{  \sum_{k\in A}\Cost^{-k}}{|A|}
    \right)
  \end{equation}

  Property (3) can be deduced from the fact $\imputation_j\geq 0$ and
  from Equation~\eqref{eq:inter:5}.\qed
\end{proof}

Those results are the technical framework that allows the detection of
coalitions on our game. Given a game $\mathcal{G}(G,\fctV,\price)$, we
can detect if a set of service providers $\SSS$ --- whose operational
costs are given by $\Cost$, but the announced prices are
$\announcedprice$ --- are currently forming a coalition.
Algorithm~\ref{algo:detection} presents the pseudo-code for the
coalition detection for a given game.

\begin{algorithm}

\KwIn{$\mathcal{G}(G,\fctV,\price)$, $\SSS$, $\Cost$, and $\announcedprice$}
\KwOut{Whether there is a coalition or not.}

compute the lowest cost path $\ShortestPath$\;

\ForAll{players $\player \in \SSS$}{compute $\price^{-\player}$}

compute $A = \{\player \in \SSS \mid c^{-\player} \le \price\}$\;

coalition = true\;
\ForAll{players $\player \in A$}{
  compute $x_\player$\;
  \If{$x_\player \ne \Cost_\player - \announcedprice_\player$}{coalition = false}%
}

\Return{coalition}

\caption{Coalition detection algorithm.}
\label{algo:detection}
\end{algorithm}

Now we will establish the relation between the price and the different costs when there is a stable imputation. Using Theorem~\ref{th:5}, we will compute the lower bound for non-empty bargaining sets.

\begin{theorem}\label{th:6}
  Let $\mathcal{G}(G,\fctV,\price)$ be a game. There exists a stable imputation $\imputation$ if and only if:
\begin{equation}
   \price \geq
      \sum_{k\in B}\Cost^{-k} -  (|B| -1 ) \Cost_{\ShortestPath}
 \end{equation}
 where $B= \{j\in B: \Cost^j = \Cost_{\ShortestPath} \land \Cost^{-j} >\Cost_{\ShortestPath}\}$.
\end{theorem}

\begin{proof}
 Let $A=\{j\in \SSS: \Cost^j \leq \price\}$. From Theorem~\ref{th:5}, we
 have:% the following property:
\[\forall j \in A,    \Cost^{-j}   \geq \frac{
     \Cost_{\ShortestPath} - \price   + \sum_{k\in A}\Cost^{-k}}{|A|}\]

 Let $B$ be a subset of $A$ such that $j\in B$ if
 $\Cost^{-j} > \Cost_{\ShortestPath}$.  Note that each player $j$
 belongs to a shortest path. Let $i$ be a player not in $B$. Thus, we
 have:

\begin{equation}
 \Cost^{-i} \geq \frac{
     \Cost_{\ShortestPath} - \price   + \sum_{k\in A}\Cost^{-k}}{|A|}
 \end{equation}

Since $\Cost^{-i} = \Cost_{\ShortestPath}$, the previous equation can be rewritten as:

\begin{equation}
  |A| \Cost_{\ShortestPath} \geq
     \Cost_{\ShortestPath} - \price   + \sum_{k\in A}\Cost^{-k}
 \end{equation}

From the definition of set $B$:
\begin{equation}
  \price   \geq
     \sum_{k\in B}\Cost^{-k} -  (|B| -1 ) \Cost_{\ShortestPath}
 \end{equation}
%This concludes the proof of Theorem~\ref{th:6}. \qed
\end{proof}

To illustrate, let us reconsider the example of Figure~\ref{fig:contre_example}. In this example, all players on set $A = \{ \Lambda, \delta, \beta, \gamma \}$ engage in the coalition.
\[
\begin{array}{llll}
\Cost_{\ShortestPath}= 6&  p=34  &\fctV(\SSS) =28 & \\
\Cost^{-\Lambda} = 12 & \Cost^{-\delta} = 6  &\Cost^{-\beta} = 6 &   \Cost^{-\gamma} = 20  \\
 \sum_{k\in A}\Cost^{-k}{|A|} = 44 & & \\
\imputation_{\Lambda} = 8 & \imputation_{\delta} = 2  &\imputation_{\beta} =2 &   \imputation_{\gamma} = 16  \\
\end{array}
\]

The unique stable imputation for this example is: $(\imputation_{\Lambda}, \imputation_{\delta} ,\imputation_{\beta} ,\imputation_{\gamma} )= (8,2,2,16)$. Note that in this example, $B=\{\Lambda,\gamma\}$,  in order to have a imputation on the bargaining set, the threshold to have a non-empty bargaining set  must be $p \geq 20$.

\section{Comparison with a truthful mechanism for lowest-cost routing}

Several works (see \cite{BGP}, for example) focus on the problem of
inter-domain routing from a mechanism-design point of view. The
mechanism-design principles applied for the routing problem is the
subject of seminal works by Nisan and Ronen~\cite{Nisan} and
Hershberger and Suri~\cite{VCG}.

Feigenbaum et al. \cite{BGP} provided a polynomial-time strategy
proof mechanism for optimal route selection in a centralized
computational model (inspired from \cite{Nisan}).  In their
formulation, the network is modeled as an abstract graph $G=(V,E)$.
Each vertex $\fctV$ of the graph is an agent and has a private type
$t_\fctV$, which represents the cost of a message transit through this
node. The mechanism-design goal is to find a lowest-cost path
$\Objection$ between two designated nodes $\U$ and $\Dest$. The
valuation of an agent $\fctV$ is $-t_\fctV$ if $\fctV$ is part of
$\Objection$ and $0$ otherwise.

Nisan and Ronen give the following simple mechanism for the problem:
the payment to agent $\fctV$ is equal to $0$ if $\fctV$ is not in $\Objection$,
and is equal to $d_{G | c_v = \infty}- d_{G | c_v = 0}$ if $\fctV$ is in
$\Objection$ where $d_{G | c_v = \alpha}$ is the cost of the
lowest-cost path through $G$ when the cost of $\fctV$ is $\alpha$.

This mechanism ensures that the dominant strategy for each agent
$\fctV$ is to always report its true type $t_\fctV$ to the
mechanism. Such a mechanism is said to be \emph{truthful}. When all
agents honestly report their costs, the lowest-cost path is selected.
%(the lowest-cost path is obtained by a simple shortest path
%calculation.)
This algorithmic mechanism design problem is solved
using the well-known VGC mechanism~\cite{Vickrey,Groves,Clarke}.

We focus on the lowest-cost routing problem: the instance is composed
of $G=(V,E)$ and a type vector $t=(t_1,\dots,t_m)$. The goal is to
find a lowest-cost path $\Objection$ between two designated nodes $\U$
and $\Dest$. We will build a coalition game
$\mathcal{G}(G,\fctV,\price)$. Each node $\fctV$ in $\SSS$ has its
cost $\Cost_\fctV$ equal to $t_\fctV$.

\begin{theorem}
  The total payment of a truthful mechanism in the lowest-cost path
  between a end-user and a destination is equal to the problem of
  finding the minimal value for $\price$, such that there exists a
  stable coalition for the choreography enactment pricing game.
\end{theorem}

\begin{proof}
  If each agent $\fctV$ reports its true type $t_\fctV$ to the
  mechanism, then a lowest-cost path $\Objection$ is chosen. The total
  payment is equal to
  $\sum\limits_{v\in \Objection} d_{G | c_v = \infty}- d_{G | c_v =
    0}$.

  By definition, $d_{G | c_v = \infty} =\Cost^{-v}$. Therefore:
\[
\begin{array}{ll}
  \sum\limits_{v\in \Objection} \left( d_{G | c_v =
    \infty}- d_{G | c_v = 0} +t_v\right) &=\sum\limits_{v\in \Objection}
\left( \Cost^{-v} -    d_{G | c_v = 0} +t_v\right) \\
& =\sum\limits_{v\in \Objection}
\Cost^{-v} -  |\Objection| \Cost_{\ShortestPath}
\end{array}
\]

The total payment is given by $\sum\limits_{v\in \Objection}
\Cost^{-v} -  (|\Objection|-1) \Cost_{\ShortestPath}$.

Note that if $v \in \Objection$, then $d_{G | c_v = 0} +t_v = \Cost^v =
\Cost_{\ShortestPath}$. Moreover, if  $d_{G | c_v = \infty} = \Cost^{-j} =
\Cost_{\ShortestPath}$, then the payment is equal to $0$.

Let $B= \{j\in \Objection\}: \Cost^j =  \Cost_{\ShortestPath} \land
\Cost^{-j} >\Cost_{\ShortestPath}\}$. Only payments given to agents in
$B$ are strictly greater than $0$. So the total payment is equal to
the minimal value such that there exists a stable coalition for game
$\mathcal{G}(G,\fctV,\price)$, where $\price=\sum\limits_{v\in \Objection} \Cost^{-v} -  |\Objection| \Cost_{\ShortestPath}$.  \qed
\end{proof}

In the general case we have shown that there exists a unique
imputation for which there is no justified objection. This imputation
exists, provided that the end-user is willing to pay a maximum price
that makes this imputation possible. This
price is exactly the total payment that agents would have received if
an auction had taken place and a truthful mechanism had been used.

\section{Related work}

The choreography enactment pricing game is similar to fair resource
allocation and networking games. Fragnelli et al.~\cite{FragnelliGM00}
studied a related cooperative game that they called the shortest path
games. Their game models agents willing to transport a good through a
network from a source to a destination. Using a graph model, and
letting agents to control the nodes, they have studied how profits
should be allocated according to the core of the cooperation.

Compared to the shortest path game, our model extends the problem by
considering that agents could possibly control several nodes --- a
vendor could offer several different services. Also, in our game, no
node can be part of all shortest paths at the same time, the opposite
of the notion of $s$-veto players introduced by Fragnelli et al. They
have focused their studies on the conditions for the existence of a
non-empty core and on the Shapley value of the game.

Maintaining the assumption of $s$-veto players, Voorneveld and
Grahn~\cite{VoorneveldG02} extended the shortest path game and proved
that the core allocations coincide with the payoff vectors in the
strong Nash equilibria of the associated non-cooperative shortest path
game.

Several subsequent papers \cite{Aziz2011,Nebel11} studied
computability and complexity aspects of this game.  Some
properties of graphs and games guaranteeing the existence of a core
have been proposed and the computability complexity of computing cores
have been established (NP-complete and \#P-complete). Other variants
(different payoffs and players controlling arcs) have been considered,
but mostly focusing on the existence and complexity of cores, whereas
this work mostly focus on the construction of the bargaining set in
polynomial time.

The flow game can be view as maximum multicommodity flow problem in a
cooperative setting. This
model can be used to identify the set of demands to satisfy and to
route this demand on the network. In this context, players own network
resources and share a capacity to deliver commodities. Kalai et
al.~\cite{Kalai1982} first considered flow games for network with a
single commodity, where a unique player owns an arc. Several studies
(for example, \cite{Derks1985,MarkakisS03,AgarwalE08}) extended this
seminal work. Those extensions encompass variations on the number of
arcs a player can control, if the player controls all or a part of the
capacity of the arc, if players control vertices, etc. Those papers
mainly focus on how to obtain the optimal flow in the network and then
on how to allocate the revenue using core allocation techniques (since
those games have non-empty cores).

%%% Local Variables:
%%% mode: latex
%%% TeX-master: paper.tex
%%% End:

\section{Conclusion}
This work presents a game-theoretic model for the problem we call the
choreography enactment pricing problem. Vendors offer different services
to users, but the lack of regulation of the market can lead to the
formation of alliances. %The model is generic enough to capture cases
%where a vendor provides more than one type of service.

% We studied the conditions for having stable coalitions for our novel
% game model.
We show that this game may have an empty core, but still has
a non-empty bargaining set. The study of the conditions that can lead
our cooperative game to a unique stable imputation resulted in a new
coalition detection algorithm.
We also show that finding the minimal user budget that leads to a
stable coalition in the choreography enactment pricing game is
equivalent to the problem of finding the total payment of a
truthful mechanism in the lowest-cost problem between a end-user and a
destination.

As future work, we are investigating the impact of allowing server
multitenancy. We are working on adding capacities to the services so
that a user can potentially enact more than one instance of the
service without paying more. This result would generalize the problem
for services on any kind of cloud computing platform.

\bibliographystyle{splncs03}
\bibliography{coalition,state-of-the-art}

\end{document}